\documentclass[12pt]{article}
\usepackage{amssymb}
\usepackage{amsfonts}
\usepackage{xcolor}

\usepackage{amssymb}
\usepackage{amsfonts}
\usepackage{amsmath}
\usepackage{fullpage}
\usepackage{setspace}
\usepackage{natbib}

\usepackage{amsfonts}
\usepackage{amssymb}
\usepackage{latexsym, amsmath, amsthm, amssymb, natbib,
   color,verbatim,subfigure,soul}
   \usepackage{xcolor}
\usepackage{graphicx}

\newcommand{\calf}{\mathcal{F}}

\newcommand{\Meet}{Meet separability}
\newcommand{\Joint}{Join separability}

\newcommand{\VG}{\Phi}
\newcommand{\vg}{\phi}

\newcommand{\VC}{\mathcal{C}}
\newcommand{\vc}{c}

\newcommand{\VH}{\Psi}
\newcommand{\vh}{\psi}

\newcommand{\VD}{\mathcal{D}}
\newcommand{\vd}{d}

\newcommand{\vt}{\tau}

\newcommand{\F}{\mathcal{F}}

\newcommand{\meetfosd}{\wedge_{\textsc{fosd}}}
\newcommand{\joinfosd}{\vee_{\textsc{fosd}}}
\newcommand{\geqfosd}{\geq_{\textsc{fosd}}}
\newcommand{\leqfosd}{\leq_{\textsc{fosd}}}

\newcommand{\maqs}{maximal adjusted quantiles}
\newcommand{\maq}{maximal adjusted quantile}
\newcommand{\minaqs}{minimal adjusted quantiles}

\newtheorem{theorem}{Theorem}

\newtheorem{corollary}{Corollary}

\newtheorem{lemma}{Lemma}

\begin{document}

\title{Multiple Adjusted Quantiles}
\author{Christopher P. Chambers and Alan D. Miller\thanks{\textbf{Chambers:} Department of Economics, Georgetown University, ICC 580, 37th and O Streets NW, Washington, D.C. 20057.  Email: Christopher.Chambers@georgetown.edu.  Web: http://chambers.georgetown.domains/. \textbf{Miller:} Faculty of Law, Western University, 1151 Richmond Street, London, Ontario N6A 3K7, Canada. Email: alan.miller@uwo.ca. Web: http://alandmiller.com/. This research was undertaken, in part, thanks to funding from the Canada Research Chairs program. Full acknowledgments to come in the final version.}}

\maketitle
\begin{abstract}We cardinally and ordinally rank distribution functions (CDFs).  We present a new class of statistics, \maqs{}, and show that a statistic is invariant with respect to cardinal shifts, preserves least upper bounds with respect to the first order stochastic dominance relation, and is lower semicontinuous if and only if it is a \maq{}.  A dual result is provided, as are ordinal results.  Preservation of least upper bounds is given several interpretations, including one that relates to changes in tax brackets, and one that relates to valuing options composed of two assets.
\end{abstract}

\section{Introduction}

Our goal is to rank real-valued CDFs cardinally, in a manner to be explained.  We do so through the use of what we call a statistic, where a statistic associates with every real-valued CDF a real number.  Our primary interest is in a property that states that, given two CDFs, the statistic of the smallest CDF which first order stochastically dominates each of them (the ``join'') coincides with the statistic of the better of the two original CDFs.  Dually, we also sometimes ask that, given two CDFs, the statistic of the largest CDF first order stochastically dominated by both of them coincides with the statistic of the worse of the two original CDFs.

We also consider a property that requires these statistics to commute with respect to cardinal shifts. We use these properties to characterize the \maqs{}, a new class of statistics that generalizes of the family of (lower) quantiles on the distribution functions.  This class of statistics has applications in a wide range of contexts, including academic grading, finance, health, and inequality.

Consider the problem of setting a tax bracket; that is, the level at which one's income becomes subject to a particular tax rate. In this context, our statistic takes as input an income distribution, and outputs the cutoff for the tax bracket. Our first axiom, \emph{meet separability}, requires the statistic to commute with respect to the meet operation on CDFs; that is, given two CDFs, the greatest CDF first order stochastically dominated by both of them must be equivalent to the lesser of the original two CDFs. To understand this axiom, consider an income distribution $F$, where $\rho(F)$ represents the tax bracket cutoff. Suppose that, as a result of the tax/transfer system or because of some exogenous event, the income distribution changes to $G$, and tax bracket cutoff falls to $\rho(G)$, so that lower levels of income are now taxed at the given rate. Our axiom states that the fall in the tax bracket cutoff must have occurred because some people's incomes decreased, and not because of those who gained. That is, the winners are irrelevant; removing their gains (resulting in distribution $F \wedge G$) would not lead to a further drop in the tax bracket cutoff (so that $\rho(F \wedge G)=\rho(G)$).

A related axiom, \emph{join separability}, requires the statistic to commute with respect to the join operation on CDFs; that is, given two CDFs, the least CDF that first order stochastically dominates both of them must be equivalent to the better of the original two CDFs. Here, a change in the income distribution that results in an increase in the tax bracket cutoff (from $G$ to $F$) must be because of the people who became richer, and not because of those made worse off. Join separability requires the statistic to ignore the losers---removing their losses (resulting in distribution $F \vee G$) would not lead to a further increase in the tax bracket cutoff (so that $\rho(F \vee G)=\rho(F)$).

Together with a translation equivariance property and a weak continuity condition, we use these axioms to characterize two classes of statistics, which we call maximal (and minimal) adjusted quantiles. The policymaker would like to base the tax bracket cutoff on the 90th percentile, but believes that she can learn something from the 80th percentile as well. Based on past experience, the policymaker believes that a reasonable difference between these percentiles is \$5000 of annual income. She calculates both and adjusts the 80th percentile by adding \$5000. Join separability allows us to characterize the \maqs{}—that is, the maximum of these two numbers—while meet separability allows us to characterize \minaqs{}, where the cutoff is their minimum. Our example involves a choice between two quantiles, but she can use any adjustment scheme that she chooses, subject to a few technical constraints.

A second application comes from finance. Consider the problem of an investor who can invest in two comonotone assets, $F$ and $G$, with statistics $\rho(F)$ and $\rho(G)$ respectively. Suppose that $\rho(F) \geq \rho(G)$. Join separability requires the statistic to ignore the states of the world in which $F$ does worse than $G$. That is, the statistic must treat $F$ as being equivalent to a derivative asset whose return tracks the better of $F$ and $G$ in all states of the world; we may think of this as an \emph{option to choose} either $F$ or $G$ after the state is realized. Similarly, meet separability requires $G$ to be equivalent to the derivative composed of the lesser of $F$ and $G$.

Our framework does not allow us to track the correlation structure between $F$ and $G$.  Implicit in our story of join separability then is that we study a statistic $\rho$ which, whenever two assets are comonotone, commutes with respect to the introduction of an option to choose either of the two.  However, the statistic must be based on the marginal distributions of the two assets alone, ignoring their correlation structure.

Using these axioms, we generalize the ``value at risk,'' a statistic used by investors to compare portfolios according to their value at a pre-specified quantile, such as the 5th percentile. In practice, investors may care about both the 5th and 10th percentiles. If the difference is small, they may be more interested in the 5th percentile, knowing that they have some ability to hedge their risk; if the difference is large, they may be more interested in the 10th percentile. Consequently, in scoring a portfolio, they adjust the 5th percentile with a constant (perhaps, the cost of hedging). We use join separability to characterize the statistic that scores the portfolio according to the maximum of two numbers, while meet separability is used to characterize their minimum.

We provide two provide two representations of \maqs{}. The easiest to understand is as follows, ignoring some technical details.  Given any CDF and any number $\alpha$ in between zero and one, we can speak of the $\alpha$-lower quantile for the CDF.\footnote{We consider lower quantiles for much of the paper: generally the set of $\alpha$ quantiles can be a non-degenerate closed interval, the lower quantile is the lowest point in this interval.} Now, for a given CDF, quantiles are weakly ordered.  So we introduce a handicap function $c$.  The function $c$ operates on a subset of numbers in between zero and one, and for a given CDF and quantile $\alpha$, we subtract this value from the CDF's $\alpha$ quantile.  We can refer to this value as the adjusted $\alpha$-lower quantile; our statistic returns the \emph{maximal} adjusted $\alpha$-quantile.

In addition to our main characterization, we present two related theorems. First, we show that a dual result holds; that is, by replacing the join separability axiom with meet separability, and appropriately modifying our continuity axiom, we establish a characterization of the class of \minaqs{} described above. Second, we provide an ordinal analogue of our results.

\subsection{Related literature}

The join separability property that we employ has been used in economics since at least \citet{kreps1979} in the context of ranking menus of available options.\footnote{Kreps refers to the property as \emph{strategic rationality}. This property has been studied in other economic contexts, including group identification \citep{miller2008,CHO2020171},  information aggregation \citep{chambers2011, dimitrov2009,cmy:2020}, efficiency measurement \citep{hk:1998,chk:1999,cm:2014a}, and influence measurement \citep{cm:2014b, cm:2018}.
For recent work on menu choice, see also \citet{ik:2021}.}
Kreps uses the property to characterize rankings that can be represented as ``indirect utilities,'' in the sense that there is an underlying ranking of alternatives such that a menu is indifferent to its highest ranked alternative.  He shows that the ordinal content of this is that if a menu $A$ is deemed as good as menu $B$, then $A$ must be indifferent to the combined menu $A \cup B$.
 In our framework, the property is essentially identical to one postulated by \citet{cgr:2008}, and we interpret it in much the same way.

The connection with \citet{kreps1979} can be made more precise. If we associate with any given CDF the closure of its supergraph, it is immediate that the union of these supergraphs corresponds to the join with respect to first order stochastic dominance of the original CDFs.  With the right technical assumptions, one can establish a Krepsian-style representation over CDFs whenever the join property is satisfied.\footnote{Similar results can also be shown to result in the framework of \citet{cgr:2008}, when restricting the join axiom to comonotonic random variables.}  Rather than establish this representation here, we instead show that by adding a translation equivariance property, we single out one particular class of statistics.

\citet{cgr:2008} characterizes a class of ``generalized quantiles'' based on a concept called the Sugeno integral, after \citet{ms:1974}.  These authors impose versions of the meet axiom and the join axiom simultaneously.  The Sugeno integral is a technical concept, a type of integral relying on non-additive set functions called capacities.  In this sense, it parallels the development of the Choquet integral in decision theory \citep{schmeidler1989}.   Essentially the same class of generalized quantiles (or Sugeno integrals) are also studied in the mathematical finance literature under the term ``Lambda Value at Risk,'' or $\Lambda V @ R$ \citep[see][]{fmp:2014,bpr:2017,cp:2018}.

Interestingly, the use of the``Sugeno integral'' in economics seems to predate \cite{cgr:2008}. \citet{hm:1980} introduces a theory of strategy-proof mechanisms for determining a public choice in a single-peaked framework.  His Proposition 3 characterizes what are essentially Sugeno integrals (with the caveat that the capacity defining the integral need neither assign value $0$ to $\varnothing$ nor $1$ to the whole set) based on strategy-proofness alone.  

Returning to sets, \citet{hk:1998,chk:1999} and \citet{cm:2014a} show that in finite-dimensional Euclidean environments, imposing both the join property and the meet property on comprehensive sets results in a \emph{monotone path} of alternatives (as in \citet{tm:1980}); the more of the path included in a given set, the higher the rank of that set.  \citet{cm:2018}  show that the reason the ``path'' emerges is because of the domain restriction of ``comprehensivity.''  In Corollary 1 there, it is shown (in a finite environment) that for other partial orders with respect to which sets may be comprehensive, the path property of ranking is preserved.

It is easy to see how to relate this to \citet{kreps1979}:  the ``path'' corresponds to a Leontief style ranking on the supergraphs (which is implied by the join property) but also a Leontief style ranking on the subgraphs (which is implied by the meet property).  In fact, Leontief style preferences are the only ones for which maximization of a relation on a set correspond (roughly) to minimization on its complement.  Similarly, the Sugeno integral is based on a ``path,'' where the path is taken with respect to the subgraph of the cdf.

Quantiles have enjoyed recent interest in the decision theory literature, starting with \citet{manski1988ordinal,chambers2007ordinal,rostek2010quantile}.  More recently, \citet{dcg:2019a,dcg:2019b} investigate quantiles in dynamic environments.  Other recent and related studies in economics include \citet{long2021,de2022people,de2022static,de2019smoothed,de2019dynamic}.  Though there has been a recent resurgence of interest, quantiles (and their generalizations) have been studied in economics for many years. In finance, quantiles (under the term ``value at risk'') are used extensively as a measure of risk.  We also remark that the idea of understanding real-valued maps of random variables is also a major theme in modern economic theory; see \emph{e.g.} \citet{mu2021monotone}.

Our characterization shares some characteristics with the ``variational'' model of decision making; and indeed, many popular models in economics.  In this literature, and ignoring risk attitudes, a state space is posited and state-contingent maps are evaluated as follows.  Each probability over the states of the world is associated with a ``cost'' function, analogous to our handicap; this cost represents the decision maker's attitude toward ambiguity.  For each probability measure and each state-contingent map, the decision maker adds the cost to the expected value of the map under that probability.  The decision maker then evaluates a state-contingent map via its minimal value across all such cost-adjusted expectations.  The model was first characterized by \citet{maccheroni2006}, and the tools of this characterization have now become standard, appearing in many contexts in economics.  To this end, we provide a dual interpretation of our result which is similar in spirit, replacing the state-contingent map and probability measure with a CDF and a quantile level.  The dualization comes from simply replacing join separability by meet separability and adjusting a continuity condition.  Moreover, we also establish the ordinal content of this result, starting from a primitive binary relation over CDFs.  We imagine this could be taken further, by starting from a state-space model and a set of state-contingent maps, endogenously determining a probability measure and utility index (as in \citet{rostek2010quantile}) and then providing a full decision theoretic characterization of the induced quantile functions from ordinal behavior.

\section{Preliminaries}

A \textbf{distribution function} on $\mathbb{R}$ is a mapping $F:\mathbb{R}\rightarrow\mathbb{R}$ which is nondecreasing, upper semicontinuous, and satisfies $\{0,1\}\subset \overline{F(\mathbb{R})}\subseteq [0,1]$.  The set of distribution functions will be denoted $\calf$.

Distribution functions can be ranked according to first order stochastic dominance; that is, $G\geqfosd F$ if for all $x\in\mathbb{R}$, $G(x)\leq F(x)$.  The stochastic dominance relation forms a lattice on the collection of distribution functions.  In particular, $G\joinfosd F$ is given by the pointwise minimum of $G$ and $F$, while $G\meetfosd F$ is given by the pointwise maximum.

For random variables $X,Y$ on a probability space $\left(\Omega,\Sigma,p\right)$, we define $X\vee Y$ and $X\wedge Y$ to be the random variables such that $(X\vee Y)(\omega)=\max\{X(\omega),Y(\omega)\}$ and $(X\wedge Y)(\omega)=\min\{X(\omega),Y(\omega)\}$. For a random variable $X$ denote by $F_X$ the associated distribution function.

We now discuss a connection between $F_{X\vee Y}$ and $F_{X}\joinfosd F_{Y}$. 
We will say random variables $X,Y$ on a probability space $\left(\Omega,\Sigma,p\right)$ are \textbf{comonotonic }if
$$\left( X\left( \omega \right) -X\left( \omega ^{\prime }\right) \right)
\left( Y\left( \omega \right) -Y\left( \omega ^{\prime }\right) \right) \geq
0\text{.}$$
Comonotonicity requires that the two random variables move weakly in the
same direction almost everywhere. It is a notion of perfect ordinal correlation.  The following result establishes that the induced lattice operation on the class of distribution functions with respect to stochastic dominance behaves nicely with respect to the pointwise lattice operation on random variables.  In particular, if two random variables are comonotonic, then the lattice operations commute with respect to taking the distribution function.  Theorem~\ref{comonotonic} is well-known but we include it here for completeness.  It permits us to relate our work to that of \citet{cgr:2008}.

\begin{theorem}
\label{comonotonic} Let $\left( \Omega ,\Sigma ,p\right) $ be a probability
space, and let $X,Y$ be random variables. Then

$$F_{X\vee Y} \geqfosd  F_{X}\joinfosd  F_{Y}$$
and
$$F_{X} \meetfosd  F_{Y} \geqfosd  F_{X\wedge Y}\text{.}$$
Moreover, if $X$ and $Y$ are comonotonic, then $F_{X\vee Y}=F_{X}\joinfosd F_{Y}$ and $F_{X\wedge Y}=F_{X}\meetfosd F_{Y}$.
\end{theorem}

The proof of Theorem \ref{comonotonic} is given in the appendix.

Clearly, the comonotonicity condition is not necessary for $F_{X\vee
Y}=F_{X}\joinfosd F_{Y}$. One way to see this is that any $X$ and $Y$ for
which for all $\omega$, $X\left( \omega \right) \geq Y\left( \omega \right) $ clearly
satisfies this equality. However, more complicated examples can be constructed.
The only point to be made here is that perfect ordinal dependence leads to the
preservation of joins (similarly, it leads to the preservation of
meets).

The following corollary demonstrates that for any two distribution
functions, there are real-valued random variables generating the distribution functions whose join
(option)\ generates the join of the distribution functions. Hence, there is
some loss of generality in working with distribution functions as opposed to
random variables, but there is also a tight relation between the
distribution function of an option and the join of the underlying
distribution functions.  This result is certainly well-known, but we state it here for completeness.

\begin{corollary}
\label{tightness}Let $F,G\in \mathcal{F}$. Then there exists a probability
space $\left( \Omega ,\Sigma ,p\right) $ and random variables $X$ and $Y$
for which $F=F_{X}$, $G=F_{Y}$ and $\left( F\meetfosd G\right) =F_{X\wedge
Y}$, $\left( F\joinfosd G\right) =F_{X\vee Y}$.
\end{corollary}

The proof of Corollary \ref{tightness} is given in the appendix.

\section{On the ranking of distribution functions}

We will call a map $\rho:\calf\rightarrow\mathbb{R}$ a \emph{statistic}.  For distribution functions $F$ over $\mathbb{R}$, there is a particularly notable statistic: the lower quantile function, where the $\alpha$ quantile for $0<\alpha<1$ is defined by
\[\rho_\alpha(F)=\inf\{x\in\mathbb{R}:F(x)\geq\alpha\}.\]

From here on out, we drop the adjective ``lower'' and simply refer to a quantile. 

We now introduce some properties that some statistics may satisfy.

The first two of these relate to the ``option'' interpretation of cdfs.  The first axiom can be interpreted as asking that, for comonotonic random variables, the option to choose one of a pair has the same value of the best of the underlying two random variables.  It was introduced in random variable form in \citet{cgr:2008}.

\begin{description} \item[\Joint:] For all $F,G\in\calf$, $\rho(F\joinfosd  G)=\max\{\rho(F),\rho(G)\}$.
\end{description}

We say that a sequence of distribution functions $\{F_n\}\subseteq\calf$ \textbf{weakly converges} to $F\in\calf$ if for all continuity points $x$ of $F$, $F_n(x)\rightarrow F(x)$.

\begin{description} \item[Lower semicontinuity:] Suppose that $\rho(F_n) \leq x^*$ for all $n$, and that $F_n \rightarrow F$.  Then $\rho(F) \leq x^*$.
\end{description}

Lastly, we define two axioms that require the statistic to commute with respect to cardinal transformations of the CDF. The first such axiom requires the statistic to require to arbitrary cardinal shifts, in which a constant value is added to each realization of the random variable.

\begin{description} \item[Translation equivariance:] For all $b\in\mathbb{R}$, if $g(x) = x+b$, then $\rho(F\circ g^{-1})=\rho(F)+b$.
\end{description}

The second such requires the statistic to commute with respect to affine transformations.

\begin{description} \item[Affine equivariance:] For all $a>0$ and  $b\in\mathbb{R}$, if $g(x)=ax+b$, then $\rho(F\circ g^{-1})=a\rho(F)+b$.
\end{description}

Our main result provides two equivalent characterizations of join separability, translation equivariance, and lower semicontinuity.

Let $\VG$ be the set of functions $\vg:\mathbb{R}\rightarrow[0,1)$ such that (a) $\vg$ is non-decreasing, (b) $\vg$ is right-continuous, (c) there exists $x\in\mathbb{R}$ such that $\vg(x)=0$, (d) there exists $x\in\mathbb{R}$ such that $\vg(x)>0$, and (e) $\lim_{x\rightarrow\infty}\vg(x)<1$.  Let $P^\VG$ be the set of statistics $\rho^\vg:\mathcal{F}\rightarrow \mathbb{R}$ such that for all $\rho^\vg\in P^\VG$, there exists $\vg\in\VG$ such that $\rho^\vg(F)=\min\{t:F(x)\geq \vg(x-t)\textnormal{ for all }x\}$.\footnote{Observe that the minimum in this expression exists; we will show that the set of $t$ for which $F(x) \geq \vg(x-t)$ for all $x$ is a closed set.  Suppose that $t^n\rightarrow t^*$, and $F(x) \geq \vg(x-t^n)$ for all $x$ and $n$.  Then $F(x+t^n) \geq \vg(x)$ for all $x$ and $n$.  Fixing $x$ and appealing to upper semicontinuity of $F$, we establish that $F(x+t^*)\geq \vg(x)$, or $F(x) \geq \vg(x-t^*)$ for all $x$, as desired.}

Let $\VC$ be the set of all non-decreasing, lower semicontinuous functions $\vc:(0,1)\rightarrow\mathbb{R}\cup\{\infty\}$ for which there exists $b\in(0,1)$ such that $c(b)=\infty$. Let $P^\VC$ be the set of functions $\rho^\vc$ such that for $\vc\in\VC$, $\rho^\vc(F)=\sup_{\alpha\in(0,1)}\left(\rho_{\alpha}(F)-c(\alpha)\right)$.\footnote{It will become clear that a function so-defined is always real valued.}

\begin{theorem}\label{3way}For a statistic $\rho :\mathcal{F}\rightarrow \mathbb{R}$, the following three statements are equivalent:
\begin{enumerate}
    \item $\rho$ satisfies join separability, translation equivariance, and lower semicontinuity.
    \item $\rho\in P^\VG$.
    \item $\rho\in P^\VC$
\end{enumerate}
\end{theorem}

The proof of Theorem~\ref{3way} is given in the appendix.

The following corollary shows that imposing affine equivariance in place of translation equivariance leads to a full characterization of quantiles.

\begin{corollary}\label{aecor} A function $\rho :\mathcal{F}\rightarrow \mathbb{R}$
satisfies join
separability, affine equivariance, and lower semicontinuity if and only if it is
a quantile for some $\alpha \in (0,1)$. \end{corollary}

The proof of Corollary ~\ref{aecor} is given in the appendix.

\section{The dual result}

The axioms of join separability and lower semicontinuity are clearly dual, in a sense, to meet separability and upper semicontinuity.  In this section, we establish formally that a dual statement to Theorem~\ref{3way} holds by imposing these dual properties.

The upper quantile function, where the $\alpha$ quantile for $0<\alpha<1$ is defined by
$$\rho_{\alpha}^+(F)=\sup\{x\in\mathbb{R}:F(x) \leq \alpha\}.$$  

We define dual analogues of the join separability and lower semicontinuity axioms.

\begin{description} \item[\Meet:] For all $F,G\in\calf$, $\rho(F\meetfosd  G)=\min\{\rho(F),\rho(G)\}$
\end{description}

\begin{description} \item[Upper semicontinuity:] Suppose that $\rho(F_n) \geq x^*$ for all $n$, and that $F_n \rightarrow F$.  Then $\rho(F) \geq x^*$.
\end{description}

Let $\VH$ be the set of functions $\vh:\mathbb{R}\rightarrow[0,1)$ such that (a) $\vh$ is non-decreasing, (b) $\vh$ is right-continuous, (c) there exists $x\in\mathbb{R}$ such that $\vh(x)=1$, (d) there exists $x\in\mathbb{R}$ such that $\vh(x)<1$, and (e) $\lim_{x\rightarrow-\infty}\vh(x)>0$.  Let $P^\VH$ be the set of statistics $\rho^\vh:\mathcal{F}\rightarrow \mathbb{R}$ such that for all $\rho^\vh\in P^\VH$, there exists $\vh\in\VH$ such that $\rho^\vh(F)=\max\{t:F(x)\leq \vh(x-t)\textnormal{ for all }x\}$.

Let $\VD$ be the set of all non-decreasing, upper semicontinuous functions $\vd:(0,1)\rightarrow\mathbb{R}\cup\{-\infty\}$ for which there exists $b\in(0,1)$ such that $\vd(b)=-\infty$. Let $P^\VD$ be the set of functions $\rho^\vd$ such that for $\vd\in\VD$, $\rho^\vd(F)=\inf_{\alpha\in(0,1)}\left(\rho^+_{\alpha}(F)-\vd(\alpha)\right)$.

\begin{theorem}\label{thm:dual}For a statistic $\rho:\mathcal{F}\rightarrow\mathbb{R}$, the following three statements are equivalent:
\begin{enumerate}
\item $\rho$ satisfies meet separability, translation equivariance, and upper semicontinuity.
\item $\rho\in P^{\VH}$
\item $\rho \in P^{\VD}$.
\end{enumerate}
\end{theorem}

The proof of Theorem~\ref{thm:dual} is given in the appendix.

\section{Ordinal comparisons}

In this section, we investigate the ordinal content of our results.  We focus on our meet separability and upper semicontinuity characterization Theorem~\ref{thm:dual}, to keep things in line with classical economic literature.  Given is a binary relation $\succeq$ on $\F$.  We want to know when there is a $\rho:\F\rightarrow\mathbb{R}$ for which for all $F,G\in\F$, $\rho(F) \geq \rho(G)$ iff $F \succeq G$, and for which $\rho$ conforms to the properties of Theorem~\ref{thm:dual}.

To simplify the exposition, denote by $F+b$ the function $F \circ h^{-1}$ where $h(x) = x+b$.  

Consider the following properties, meant to invoke the content of the cardinal properties in Theorem~\ref{thm:dual}

The first property is a known necessary condition for any ranking which can be represented by a real-valued function:

\begin{description} \item[Weak Order:] $\succeq$ is complete and transitive.
\end{description}

Our first structural axiom is intended to capture the ordinal content of \Meet.  

\begin{description} \item[Ordinal Meet Separability] If $G \succeq F$, then $(F \meetfosd G) \sim F$.
\end{description}

Our next property is an ordinal counterpart of translation equivariance.

\begin{description} \item[Ordinal Translation Invariance]  $F \succeq G$ implies $F + b \succeq G + b$.
\end{description}

We need an analogue of the continuity condition. The following is standard:

\begin{description} \item[Ordinal Upper Semicontinuity:] If for all $n$, $F_n \succeq G$ and $F_n \rightarrow F$, then $F\succeq G$.
\end{description}

Finally, we postulate an axiom claiming that any $F$ can be strictly dominated below by a degenerate distribution function.

For $x\in\mathbb{R}$, we define

\[F_x(y)=\left\{
	\begin{array}{ll}
		0  & \mbox{if } y < x \\
		1 & \mbox{if } y \geq x
	\end{array}
\right. .\]

$F_x$ represents a point mass at $x$. 

\begin{description} \item[Lower Unboundedness:] For all $F$, there exists $x$ for which $F \succ F_x$.
\end{description}

We present an ordinal analogue of Theorem~\ref{thm:dual}.

\begin{theorem}\label{thm:ordinal}A binary relation $\succeq$ satisfies weak order, ordinal meet separability, ordinal translation invariance, ordinal upper semicontinuity, and lower unboundedness if and only if there is $\rho:\F\rightarrow\mathbb{R}$ for which $F\succeq G$ if and only if $\rho(F)\geq \rho(G)$ and $\rho$ satisfies the properties of Theorem~\ref{thm:dual}.\end{theorem}

The proof of Theorem~\ref{thm:ordinal} is given in the appendix.

\section{Conclusion}

We have presented a new class of statistics, \maqs{}, and have shown that they are characterized by three axioms: join separability, translation equivariance, and lower semicontinuity. We have also provided dual and ordinal results. 

The statistics we introduce we provide have many potential uses. We described two applications in the introduction. The first involved the setting of tax brackets. More generally, we may apply these ideas in other contexts where we use a distribution to select a cutoff point. For example, when grading exams, professors must choose cutoff points that determine whether a student's score is sufficient to get an A, or to pass an exam. Similarly, an government agency may wish to set a poverty line to determine whether an individual qualifies as poor.

The second application we provided involved the ranking of financial assets. In this context, our statistic is a generalization of a commonly used measure known as value at risk. We may also apply this idea to other contexts in which a researcher may need to cardinally score or ordinally rank distributions, of incomes, wealth, health outcomes, student performance, weather, etc.

\section*{Appendix}

\begin{proof}[Proof of Theorem \ref{comonotonic}.] As $X\vee Y\geq X$ and $X\vee Y\geq Y$, it is clear that
$F_{X\vee Y}\geqfosd F_{X}$ and $F_{X\vee Y}\geqfosd F_{Y}$.
Consequently, $F_{X\vee Y}\geqfosd F_{X}\joinfosd F_{Y}$. A similar
statement holds for meets.

Now, suppose that $X$ and $Y$ are comonotonic. We will establish that $F_{X\vee Y}=F_{X}\joinfosd F_{Y}$. By means of contradiction, suppose
there exists $z\in \mathbb{R}$ for which $F_{X\vee Y}\left( z\right) <\min
\left\{ F_{X}\left( z\right) ,F_{Y}\left( z\right) \right\} $. Then

$$p\left( \left\{ \omega :X\left( \omega \right) \vee Y\left( \omega \right)
\leq z\right\} \right) <p\left( \left\{ \omega :X\left( \omega \right) \leq
z\right\} \right) \text{,}$$

so there exists some $\omega \in \Omega $ where $X\left(\omega\right) \leq
z<Y\left(\omega\right) $. Similarly, there exists $\omega^{\prime}\in
\Omega $ for which $Y\left(\omega^{\prime }\right) \leq z<X\left(\omega
^{\prime}\right)$. Now, if $Y\left(\omega\right) \leq Y\left( \omega
^{\prime}\right)$, we conclude that $z<Y\left( \omega \right) \leq Y\left(
\omega^{\prime}\right) \leq z$, or $z<z$, a contradiction. Hence, $Y\left(
\omega^{\prime}\right) <Y\left( \omega \right)$. Hence, by comonotonicity
of $X$ and $Y$, $X\left(\omega^{\prime}\right) \leq X\left(\omega
\right) $. Thus, $z<X\left(\omega^{\prime}\right) \leq X\left(\omega
\right) \leq z$, or $z<z$, a contradiction. Therefore, 
$F_{X\vee Y}\left(z\right) =\min\left\{F_{X}\left(z\right),F_{Y}\left(z\right)\right\}$. 
A similar proof establishes the meet equality.\end{proof}

\begin{proof}[Proof of Corollary \ref{tightness}.] Let $F$,$G\in \mathcal{F}$. As usual define the
point-valued inverses on $\left[0,1\right]$ as
\begin{eqnarray*}
F^{-1}\left( y\right)  &=&\inf \left\{ x:F\left( x\right) \geq y\right\}  \\
G^{-1}\left( y\right)  &=&\inf \left\{ x:G\left( x\right) \geq y\right\}
\text{.}
\end{eqnarray*}
Consider the probability space $\left(\left(0,1\right),\Sigma,p\right)$, 
where $\Sigma$ are the Lebesgue measurable sets and $p$ is the uniform
distribution on $\left(0,1\right) $. Let $X=F^{-1}$ and $Y=G^{-1}$. Note
that $F_{X}=F$ and $F_{Y}=G$.\footnote{For example, we can show that $F^{-1}(y) \leq z$ if and only if $y \leq F(z)$.  Hence, the probability of $y$ for which $F^{-1}(y)\leq z$ is $F(z)$.  To see this, if $y\leq F(z)$, then by definition $F^{-1}(y)\leq z$.  Suppose instead that $y> F(z)$.  By right continuity and nondecreasingness of $F$, there is $\epsilon^*>0$ for which for all $0\leq \epsilon\leq \epsilon^*$, $y>F(z+\epsilon^*)$.  Consequently, by definition $F^{-1}(y)>z$. These two facts establish the claim.} Moreover, $X$ and $Y$ are both weakly
increasing, hence comonotonic. Therefore, by Theorem \ref{comonotonic}, the conclusion holds.\end{proof}

We state and prove the following three lemmas.

\begin{lemma}\label{lem:monfosd}A statistic $\rho$ satisfies join separability only if, for all $F,G$, $F\geqfosd  G$ implies that $\rho(F) \geq \rho(G)$.\end{lemma}

\begin{proof}[Proof of Lemma~\ref{lem:monfosd}.]
Let $\rho$ satisfy join separability and let $F,G\in\mathcal{F}$ such that $F\geqfosd  G$. By join separability, $\rho(F) \vee \rho(G) = \rho(F \joinfosd G) = \rho(F)$. Because $\rho(F) \vee \rho(G) \geq \rho(G)$ it follows that $\rho(F) \geq \rho(G)$.\end{proof}

\begin{lemma}\label{lemmaa} Let $g:\mathbb{R}\rightarrow (0,1]$ be any nonincreasing function, for which there is $x\in\mathbb{R}$ such that $g(x) < 1$.  Then $\rho(F) =\inf\{x:F(x)\geq g(x)\}$ is lower semicontinuous.\end{lemma}

\begin{proof}[Proof of Lemma \ref{lemmaa}.] Let $F_n \rightarrow  F$, and suppose that $\rho(F_n) \leq x^*$ for all $n$.  Then by definition, $F_n(x^*+\epsilon) \geq g(x^*+\epsilon)$ for all $n$ and $\epsilon > 0$.  Let $x_m\rightarrow x^*$ be any sequence of continuity points of $F$ for which $x_m > x^*$ (the continuity points of $F$ are dense in $\mathbb{R}$).  Then since $F_n$ is nondecreasing and $g$ is nonincreasing, $F_n(x_m) \geq g(x_m)$ for all $m$, and since $F_n\rightarrow F$, by taking limits we obtain $F(x_m)\geq g(x_m)$ for all $m$.  Therefore, $\inf\{x:F(x) \geq g(x)\}\leq x^*$.  \end{proof}

\begin{lemma}\label{lemmab} Let $g:\mathbb{R}\rightarrow (0,1]$ be any nonincreasing function, for which there is $x\in\mathbb{R}$ such that $g(x)<1$.  Then $\rho(F)=\inf\{x:F(x) \geq g(x)\}$ satisfies meet and join separability.\end{lemma}

\begin{proof}[Proof of Lemma \ref{lemmab}.]Let $F,G\in \calf$ and suppose that $\rho(F)\geq \rho(G)$.  We first consider $F\joinfosd G$, and let $x^*=\rho(F)$.  Observe that for every $\epsilon >0$, $F(x^*-\epsilon)<g(x^*-\epsilon)$.  Therefore, for any $\epsilon >0$, $(F\joinfosd G)(x^*-\epsilon)<g(x^*-\epsilon)$, so that $\rho(F\joinfosd G)\geq x^*$.  On the other hand, for any $\epsilon > 0$, we have $F(x^*+\epsilon)\geq g(x^*+\epsilon)$, similarly since $\rho(G)\leq x^*$, we have $G(x^*+\epsilon)\geq g(x^*+\epsilon)$, so that $(F\joinfosd G)(x^*+\epsilon) \geq g(x^*+\epsilon)$, so that $\rho(F\joinfosd G)\leq x^*$, implying $\rho(F\joinfosd G)=x^*$.

Next we consider $F\meetfosd G$.  Analogously, let $y^*=\rho(G)\leq  \rho(F)$.  Then for any $\epsilon>0$, $G(y^*+\epsilon)\geq g(y^*+\epsilon)$, so that $(F\meetfosd G)(y^*+\epsilon)\geq g(y^*+\epsilon)$, so that $\rho(F\meetfosd G) \leq y^*$.  For any  $\epsilon > 0$, $G(y^*-\epsilon)< g(y^*-\epsilon)$ and since $y^*\leq \rho(F)$, $F(y^*-\epsilon)<g(y^*-\epsilon)$, so that $(F\meetfosd G)(y^*-\epsilon)<g(y^*-\epsilon)$, which implies that $\rho(F\meetfosd G)\geq y^*$, or $\rho(F\meetfosd G)=y^*$. \end{proof}

For the following proofs, we define, for $x\in\mathbb{R}$, the distribution $F_x$, where 

\[F_x(y)=\left\{
	\begin{array}{ll}
		0  & \mbox{if } y < x \\
		1 & \mbox{if } y \geq x
	\end{array}
\right. .\]

\begin{proof}[Proof of Theorem~\ref{3way}] \textbf{Part 1:} We show that statement (1) implies statement (2). Let $\rho$ satisfy the three axioms. For $z\in\mathbb{R}$, define $\vg_z(x)\equiv\inf\{F(x):\rho(F)\leq z\}$. Denote $\vg\equiv \vg_0$.

\textbf{Step One.} We first show that $\vg_{z}(x)=\vg(x-z)$ for all $x\in\mathbb{R}$. Let $z\in\mathbb{R}$ and let $h(x)=x-z$. Recall that $\vg_{z}(x)=\inf\{F(x):\rho(F)\leq z\}$. By translation equivariance, $\rho(F\circ h^{-1})=\rho(F)-z$ and therefore $\rho(F)=\rho(F\circ h^{-1})+z$. It follows that $\vg_{z}(x)=\inf\{F(x):\rho(F\circ h^{-1})\leq 0\}=\inf\{(F\circ h)(x):\rho(F)\leq 0\}=\inf\{F(x-z):\rho(F)\leq 0\}=\vg(x-z)$.

\textbf{Step Two.} We demonstrate that $\vg$ is in $\VG$; that is, that it satisfies the five properties.

To see that $\vg$ is non-decreasing, let $x \leq x'$.  Then for all $F$ for  which $\rho(F)\leq 0$, we have $F(x) \leq F(x')$.  Obviously then $\vg(x) \leq F(x')$ for all $\rho(F)\leq 0$. It follows that $\vg(x)\leq\inf\{F(x'):\rho(F)\leq 0\}=\vg(x')$.

To see that $\vg$ is right-continuous, note that all functions in $\calf$ are upper-semicontinuous and non-decreasing. Because $\vg$ is the pointwise infimum of a collection of upper-semicontinuous functions, $\vg$ must be upper-semicontinuous \citep[See][Page 76, Theorem 3]{berge}. Because $\vg$ is upper-semicontinuous and non-decreasing, it follows that $\vg$ is right-continuous.

To see that there exists $x\in\mathbb{R}$ such that $\vg(x)=0$, consider the distribution $F_0$, where $F_0(x)=0$ for $x<0$ and $F_0(x)=1$ for $x\geq 0$.  By construction, $F_0(x) \geq \vg_{\rho(F_0)}(x)$ for all $x$. Therefore $F_0(x)\geq \vg(x-\rho(F_0(x)))$ for all $x$. This implies that $F_0(x+\rho(F_0(x)))\geq \vg(x)$ for all $x$.  Thus $\vg(x)=0$ for all $x < -\rho(F_0(x))$.

To see that there exists $x\in\mathbb{R}$ such that $\vg(x)>0$, suppose by means of contradiction that $\vg(x)=0$ for all $x$. We show that $\rho(F_0)\leq z$ for all $z\in\mathbb{R}$, a contradiction which proves the claim.
Let $z\in\mathbb{R}$. From step one, it follows that $\vg_z(x)=0$ for all $x\in\mathbb{R}$.

We claim that there exists $F^\ast\in\mathcal{F}$ such that $\rho(F^\ast)\leq z$ and $F^\ast(0)=0$. 
Since $\vg_{z}(0) = 0$, there is a sequence $F^n \in\mathcal{F}$ for which $\rho(F^n) \leq z$ and $F^n(0) \rightarrow 0$.
Consider the distribution $F_1$, where $F_1(x)=0$ for $x<1$ and $F_1(x)=1$ for $x \geq 1$.  Let us define $G^n = F^n \meetfosd F_1$.  
Now, observe that $F^n \geqfosd  G^n$. By Lemma~\ref{lem:monfosd}, it follows from join separability that $\rho(F^n) \geq \rho(G^n)$. Consequently, $\rho(G^n)\leq z$.
Next, define $H^1=G^1$ and $H^n = G^n \joinfosd  H^{n-1}$ for $n\geq 2$. By join-separability, $\rho(H^n) \leq z$ for each $n$; further, for each $x$, $H^{n+1}(x) \leq H^{n}(x)$ for all $n$.
 Finally, it is easy to see that $H^n \rightarrow F^\ast$ for some $F^\ast\in\mathcal{F}$. 
 That is, take $F^\ast(x) = \inf_n H^n(x)= \lim_{n\rightarrow\infty} H^n(x)$, and observe that $F^\ast\in\mathcal{F}$:  this follows as $F^\ast$ is nondecreasing, satisfies the appropriate limit conditions, and is also upper semicontinuous as the lower envelope of upper semicontinuous functions (by the same \citep{berge} reference as above).
Obviously $F^\ast(0) = 0$.  By lower semicontinuity, $\rho(F^\ast)\leq z$. 

For the distribution $F_0$, $F_0(x)=1$ for $x\geq 0$; therefore $F_0(x)\geq F^\ast(x)$ for all $x$. This implies that $F^\ast \geqfosd  F_0$. Thus by Lemma~\ref{lem:monfosd} $\rho(F^\ast)\geq\rho(F_0)$ and therefore $\rho(F_0)\leq z$.

To show that $\lim_{x\rightarrow\infty}\vg(x)<1$, 
suppose by means of contradiction that $\lim_{x\rightarrow\infty}\vg(x)=1$. First, suppose there exists $x^*$ for which $\vg(x^*)=1$.
Let $F\in\mathcal{F}$ for which $F(x)<1$ for all $x$.
Define $h:\mathbb{R}\rightarrow \mathbb{R}$ by $h(x)=x-\rho(F)$.
Because $(F\circ h^{-1})(x^*)<1$, $\rho(F\circ h^{-1})>0$.  But $\rho(F\circ h^{-1})=\rho(F)-\rho(F)$ by translation equivariance, obtaining $0>0$, a contradiction.

Next, suppose there is no such $x^*$.  Observe then that $\vg\in \mathcal{F}$. Let $h(x) = \sqrt[3]{x}$ and note that $F^*\equiv\vg\circ h\in\mathcal{F}$ as well.
Observe that there is $\overline{x}$ such that for all $x\geq\overline{x}$, $x-\rho(F^*)>h(x)$ and such that $x-\rho(F^*)-h(x)$  is increasing in $x$.
Let $\varepsilon>0$ such that $\varepsilon<\overline{x}-\rho(F^*)-h(\overline{x})$.
Let $x^*\geq\overline{x}$ such that $\vg(h(x^*)+\varepsilon)>\vg(h(x^*))$.\footnote{The existence of $x^*$ follows as $\vg(x)<1$ for every $x$ and $\lim_{x\rightarrow\infty}\vg(x)=1$.}
Note that, by construction, $x^*-\rho(F^*)>h(x^*)+\varepsilon$.
Consequently, by monotonicity, $\vg(x^*-\rho(F^*))\geq\vg(h(x^*)+\varepsilon)$.
Therefore, $\vg(x^*-\rho(F^*))>\vg(h(x^*))=F^*(x^*)$.  This implies that $\vg_{\rho(F^*)}(x^*)>F^*(x^*)$. But, by the definition of $\vg$,  $F^*(x) \geq \vg_{\rho(F^*)}(x)$ for all $x$, a contradiction.

\textbf{Step Three.} Let $F^*\in\mathcal{F}$ such that $F^*(x)\geq \vg(x)$ for all $x$. We show that $\rho(F^*)\leq 0$. For a function $f:\mathbb{R}\rightarrow\mathbb{R}$, define $U(f)=\{(x,y)\in\mathbb{R}^2:y>f(x)\}$ to be the complement in $\mathbb{R}^2$ of the subgraph of $f$. Because $U(\vg)$ is open in $\mathbb{R}^2$, it is separable and is therefore a Lindel\"{o}f space \citep[see][Theorem 16.11]{willard}. Because $U(F)$ is open in $U(\vg)$ for all $F$ such that $\rho(F)\leq 0$, and because $U(\vg)\subseteq\bigcup_{F\in\calf:\rho(F)\leq 0}U(F)$, it follows that $\{U(F):\rho(F)\leq 0\}$ is an open cover of $U(\vg)$. Thus, by the definition of a Lindel\"{o}f space, there is a countable subcover $\{F_n\}\subset\{F:\rho(F)\leq 0\}$ so that $U(\vg)\subseteq\bigcup_n U(F_n)$. Since for each $n$, $U(F_n)\subseteq U(\vg)$, we may conclude that $U(\vg) = \bigcup_n U(F_n)$.  Consequently \[\{(x,y)\in\mathbb{R}^2:\vg(x) \geq y\}=\bigcap_n \{(x,y)\in\mathbb{R}^2:F_n(x)\geq y\}.\]  

Now define $F_1'(x)=F_1(x)$ and inductively for $n\geq 2$, $F_n'(x)=\min\{F_n(x),F'_{n-1}(x)\}$.  For each $n$, $\rho(F_n) \leq 0$, so conclude by induction and by join separability that $\rho(F'_n)\leq 0$.  Then for all $n\geq 1$ and all $x$, $F'_{n+1}(x)\leq F'_n(x)$ and observe that \begin{equation}\label{eq:intersection}\{(x,y)\in\mathbb{R}^2:\vg(x) \geq y\}=\bigcap_n \{(x,y)\in\mathbb{R}^2:F'_n(x)\geq y\}.\end{equation}

Finally, we claim that $(F'_n \meetfosd  F^*)\rightarrow F^*$.  For all $x\in\mathbb{R}$, since $F_n'(x)$ is a decreasing sequence in $\mathbb{R}$, $\lim_n F_n'(x)=\vg(x)$.  As $F^*(x)\geq \vg(x)$ for all $x$, conclude that $\lim_n \max\{F_n'(x),F^*(x)\}=F^*(x)$.  Consequently $(F_n'\meetfosd  F^*)\rightarrow F^*$. 

Because  $(F'_n\meetfosd F^*) \leqfosd F'_n$, it follows from join separability and Lemma~\ref{lem:monfosd} that $\rho(F'_n \meetfosd  F^*)\leq \rho(F'_n) \leq 0$ for all $n$.
 By lower semicontinuity, because $\rho(F'_n \meetfosd  F^*)\leq 0$ for all $n$ and $(F'_n \meetfosd  F^*)\rightarrow F^*$, it follows that $\rho(F^*)\leq 0$.

\textbf{Step Four.} Let $F\in\mathcal{F}$. We show that $\rho(F)=\min\{t:F(x)\geq \vg(x-t)\textnormal{ for all }x\}$. By construction, $F(x)\geq\vg_{\rho(F)}(x)=\vg(x-\rho(F))$ for all $x\in\mathbb{R}$. This implies that $\rho(F)\in\{t:F(x)\geq \vg(x-t)\textnormal{ for all }x\}$ and therefore that $\min\{t:F(x)\geq \vg(x-t)\textnormal{ for all }x\}\leq\rho(F)$. 
Next, let $t\in\mathbb{R}$ such that $F(x)\geq\vg(x-t)$ for all $x\in\mathbb{R}$. Define $h(x)=x-t$, and note that  $F(x+t)=(F \circ h^{-1})(x)\geq\vg(x)$ for all $x\in\mathbb{R}$. By step three, $\rho(F \circ h^{-1})\leq 0$, which implies by translation equivariance that $\rho(F)\leq t$. Thus $\rho(F)\leq\min\{t:F(x)\geq \vg(x-t)\textnormal{ for all }x\}$. Consequently, $\rho(F)=\min\{t:F(x)\geq \vg(x-t)\textnormal{ for all }x\}$.

\textbf{Part 2:} We show that statement (2) implies statement (3). Let $\rho\in P^\VG$, and let $\vg\in\VG$ such that $\rho=\rho^\vg$. 
Define $\vc:(0,1)\rightarrow\mathbb{R}\cup\{+\infty\}$ such that $\vc(\alpha)=\inf\{x:\vg(x)\geq\alpha\}$, where an infimum over the empty set is defined as $+\infty$. We show that $\vc\in\VC$ and $\rho=\rho^\vc$.

To show that $\vc$ is non-decreasing, let $\alpha,\alpha'\in(0,1)$ such that $\alpha < \alpha'$. The fact that $\alpha<\alpha'$ implies that $\{x:\vg(x)\geq\alpha'\}\subseteq\{x:\vg(x)\geq\alpha\}$.
Therefore, $\inf\{x:\vg(x)\geq\alpha\}\leq\inf\{x:\vg(x)\geq\alpha'\}$, and therefore that $\vc(\alpha)\leq\vc(\alpha')$.

To show that $\vc$ is lower semicontinuous, we show that its supergraph is closed. Because $\vg$ is upper-semicontinuous, its subgraph is closed. Thus we show that a point $(x,y)\in\mathbb{R}\times(0,1)$ is in the supergraph of $\vc$ if and only if $(y,x)$ is in the subgraph of $\vg$. 
Let $(x,\alpha)\in\mathbb{R}\times(0,1)$ such that $\vg(x)\geq \alpha$. By the definition of $\vc$, this implies that $x\geq\vc(\alpha)$. 
Next, let $(x,\alpha)\in\mathbb{R}\times(0,1)$ such that $x\geq\vc(\alpha)$. Note that, as the subgraph of $\vg$ is closed, $c(\alpha)=\inf\{x:\vg(x)\geq\alpha\}\in\{x:\vg(x)\geq\alpha\}$. Consequently, $\vg(x)\geq \alpha$.

To show that there exists $b\in(0,1)$ such that $\vc(b)=+\infty$, let $b\in(0,1)$ such that $b>\lim_{x\rightarrow\infty}\vg(x)$. By construction, the set $\{x:\vg(x)\geq b\}$ is empty. Hence $\vc(b)=\inf\{x:\vg(x)\geq b\}=+\infty$.

Finally we claim that for all $F\in\mathcal{F}$, $\rho(F) = \sup_{\alpha\in(0,1)}\left(\rho_{\alpha}(F)-\vc(\alpha)\right)$.  First, we establish that $\rho(F) \geq \sup_{\alpha\in(0,1)}\left(\rho_{\alpha}(F)-\vc(\alpha)\right)$.  Let $\alpha \in (0,1)$ be arbitrary, and choose any $t\in\mathbb{R}$ for which for all $y$, $F(y+t)\geq \vg(y)$.  Because of upper semicontinuity of $\vg$, we have $\vg(\vc(\alpha))\geq \alpha$.  Because of the choice of $t$, we have $F(\vc(\alpha)+t)\geq\vg(c(\alpha))\geq \alpha$.  By definition of $\rho_{\alpha}(F)$, we then may conclude that $\rho_{\alpha}(F)\leq \vc(\alpha)+t$.  So, for any $t$ for which $F(y+t)\geq \vg(y)$ for all $y$, we have $\rho_{\alpha}(F)\leq \vc(\alpha)+t$.  Conclude that $\rho_{\alpha}(F)\leq \vc(\alpha)+\inf\{t:F(y+t)\geq \vg(y)\mbox{ for all }y\}=\vc(\alpha)+\rho(F)$.  Since $\alpha$ was arbitrary, $\sup_{\alpha\in(0,1)}\rho_{\alpha}(F)-\vc(\alpha)\leq\rho(F)$, which is what we wanted to show.

To show equality with the supremum, let us assume by means of contradiction that there is some $F\in\mathcal{F}$ and $\hat{\epsilon}>0$ for which for all $\alpha\in(0,1)$ $\rho(F)-\hat{\epsilon}\geq \rho_{\alpha}(F)-\vc(\alpha)$.  Without loss, we may assume that $\rho(F)=0$, so that $-\hat{\epsilon}\geq\rho_{\alpha}(F)-\vc(\alpha)$.

By definition of $\rho$ and since $\rho(F)=0$, for every $\epsilon  > 0$, there exists an $x(\epsilon)$ so that $\vg(x(\epsilon)+\epsilon)>F(x(\epsilon))$.  Let  $\hat{\alpha}=\phi(x(\hat{\epsilon})+\hat{\epsilon})$.  Since $\vg(x(\hat{\epsilon})+\hat{\epsilon})>F(x(\hat{\epsilon}))$, we know by upper semicontinuity of $F$ that $\rho_{\hat{\alpha}}(F)>x(\hat{\epsilon})$.  Since $\vg(x(\hat{\epsilon})+\hat{\epsilon})=\hat{\alpha}$, we know that $\vc(\hat{\alpha})\leq x(\hat{\epsilon})+\hat{\epsilon}$.  So:  $\rho_{\hat{\alpha}}(F)-\vc(\hat{\alpha})> x(\hat{\epsilon})-(x(\hat{\epsilon})+\hat{\epsilon})=-\hat{\epsilon}$, a contradiction as we assumed $-\hat{\epsilon}\geq \rho_{\alpha}(F)-\vc({\alpha})$ for every $\alpha\in(0,1)$.

\textbf{Part 3:} We show that statement (3) implies statement (1). Let $\rho\in P^\VC$. We show that $\rho$ satisfies the three axioms. Let $\vc\in\VC$ such that $\rho=\rho^\vc$.

This proof makes use of the fact that $\rho_\alpha$ satisfies join separability, affine equivariance, and lower semicontinuity for all $\alpha\in(0,1)$. That $\rho_\alpha$ satisfies affine equivariance is obvious. That $\rho_\alpha$ satisfies lower semicontinuity and join separability follows from Lemmas~\ref{lemmaa}~and~\ref{lemmab}.

To show that $\rho$ satisfies join separability, let $F,G\in\mathcal{F}$.  For $\alpha\in(0,1)$, the fact that $\rho_\alpha$ satisfies join separability implies that $\left(\rho_{\alpha}(F \joinfosd  G)-\vc(\alpha)\right)=\max_{E\in\{F,G\}}\left(\rho_{\alpha}(E)-\vc(\alpha)\right)$.
 Consequently, $\rho(F \joinfosd  G)=\sup_{\alpha\in(0,1)}\max_{E\in\{F,G\}}\left(\rho_{\alpha}(E)-c(\alpha)\right)=\max_{E\in\{F,G\}}\sup_{\alpha\in(0,1)}\left(\rho_{\alpha}(E)-\vc(\alpha)\right)=\max\{\rho(F),\rho(G)\}$.

To show that $\rho$ satisfies translation equivariance, let $b\in\mathbb{R}$ and $F\in\mathcal{F}$, and define $h(x)=x+b$. 
Because $\rho_\alpha$ satisfies affine equivariance, $\rho_{\alpha}(F\circ h^{-1})=\rho_{\alpha}(F)+b$.
Consequently, $\rho(F\circ h^{-1})=\sup_{\alpha\in(0,1)}\left(\rho_{\alpha}(F\circ h^{-1})-\vc(\alpha)\right) = \sup_{\alpha\in(0,1)}\left(\rho_{\alpha}(F)+b-\vc(\alpha)\right)=\sup_{\alpha\in(0,1)}\left(\rho_{\alpha}(F)-\vc(\alpha)\right)+b$. 
Consequently, $\rho(F\circ h^{-1})=\rho(F)+b$. 

To show that $\rho$ satisfies lower semicontinuity, let $F\in\mathcal{F}$, let $z\in\mathbb{R}$, and let $(F_n)$ be a sequence of CDFs  such that (a) $\rho(F_n)\leq z$ for all $n$ and (b) $F_n \rightarrow F$.  We show that $\rho(F)\leq z$.
Because $\rho(F_n) = \sup_{\alpha\in(0,1)} \rho_\alpha(F_n) - \vc(\alpha) \leq z$ for all $n$, it follows that $\rho_{\alpha}(F_n)\leq z+\vc(\alpha)$ for all $n$ and all $\alpha$. 
Because $\rho_\alpha$ satisfies lower semicontinuity, the facts that $F_n \rightarrow F$ and $\rho_{\alpha}(F_n)\leq z+\vc(\alpha)$ imply that $\rho_{\alpha}(F)\leq z+\vc(\alpha)$ and therefore that $\rho_{\alpha}(F)-\vc(\alpha)\leq z$.
Consequently, $\sup_{\alpha\in(0,1)}\left(\rho_{\alpha}(F)-\vc(\alpha)\right)=\rho(F)\leq z$.\end{proof}

\begin{proof}[Proof of Corollary~\ref{aecor}]
First, we show that quantiles satisfy affine equivariance. To see this, let $F\in\mathcal{F}$ and let $\rho_\alpha$ be a quantile for some $\alpha\in(0,1)$. Recall that $\rho(F)=\inf\{x:F(x)\geq\alpha\}$. Let $a>0$ and $b\in\mathcal{R}$ and define $h(x)=ax+b$. Then $\rho(F\circ h^{-1})=\inf\{x:(F\circ h^{-1})(x)\geq\alpha\}=\inf\{h(x):F(x)\geq\alpha\}=h(\rho(F))=a\rho(F)+b$.

Next, let $\rho$ satisfy the three axioms. Because $\rho$ satisfies affine equivariance, it satisfies transnslation equivariance. Thus there is $\vg\in\VG$ such that $\rho=\rho^\vg$. For $\alpha\in(0,1)$, define the function $\vg^\alpha$ such that $\vg^\alpha(x)=0$ for $x<0$, and $\vg^\alpha(x)=\alpha$ otherwise. By the definition, $\rho_\alpha(F)=\rho^{\vg^\alpha}(F)$. We show that there is $\alpha\in(0,1)$ such that $\vg=\vg^\alpha$.

Define $\alpha=\vg(0)$, and suppose by means of contradiction that $\vg\neq\vg^\alpha$. 
There are two cases. First, suppose $x'=\inf\{x:\vg(x)>0\}\neq 0$. For all $x\in\mathbb{R}$, $\rho(F_x)=x-x'$. Let $a>1$ and define $h(x)=ax$.  Note that $F_x\circ h^{-1}=F_{ax}$. By affine equivariance, $a \rho(F_x)=\rho(F_{ax})$. Then $0=a \rho(F_{x'})=\rho(F_{ax'})=ax'-x'\neq 0$, a contradiction. Because $\vg$ is upper-semicontinuous and bounded away from 1, it follows that $\alpha\in(0,1)$.

Next, suppose there exists $x'>0$ such that $\vg(x')>\alpha$. 
By upper semicontinuity of $\phi$, there exists $y<x'$ such that $y>0$ and $\vg(y)<\vg(x')$.
Let $\beta > 0$ such that $\beta y > x'$.
Define $F(x)\equiv\vg(x)$ for $x<3x'$, and $F(x)\equiv 1$ for $x\geq 3x'$.
Note that $F\in\mathcal{F}$.
Because $F(x) \geq \phi(x)$ for all $x$, it follows that $\rho(F)\leq 0$.
Next, let $h(x) = \beta x$. 
Because $F \circ h^{-1} (\beta y) < \phi(\beta y)$, it follows that $\rho(F \circ h^{-1}) > 0$.
But by affine equivariance, $\rho(F \circ h^{-1}) = h(\rho(F)) = \beta\rho(F) \leq 0$, a contradiction.
\end{proof}

\begin{proof}[Proof of Theorem~\ref{thm:dual}]
\textbf{Part 1.}
Consider the function $\vt:\mathcal{F}\rightarrow\mathcal{F}$ defined as follows.  Let $X$ be any random variable $X$ for which $F$ is the cdf, and let $\vt(F)\equiv F_{-X}$ be the cdf of $-X$.  Observe that by the Theorem 2.7 of \citet{billingsley} and Example 2.3 of the same, we have $F_n \rightarrow F$ if and only if $\vt(F_n) \rightarrow \vt(F)$.  Also observe that $\vt(F) \meetfosd \vt(G)=\vt(F\joinfosd G)$ {and $\vt(F) \joinfosd \vt(G)=\vt(F \meetfosd G)$}.\footnote{The lattice result follows from the fact that $F\geqfosd G$ if and only if $\vt(G)\geqfosd \vt(F)$.  For, suppose $F\geqfosd G$, and observe that $\vt(F)(x)=1-F(-x^-)\geq 1-G(-x^-)=\vt(G)(x)$, so that $\vt(G)\geqfosd \vt(F)$.  And suppose that $\vt(G)\geqfosd \vt(F)$, but that $F\geqfosd G$ is false.  Then there is $x$ for which $F(x)>G(x)$, and in particular by right continuity, there is $\varepsilon^* > 0$ and $a>b$ such that $F(x+\varepsilon)\geq a > b \geq G(x+\varepsilon)$ for all $0<\varepsilon<\varepsilon^*$.  Then $\vt(F)(-(x+\varepsilon^*))\leq 1-a<1-b \leq \vt(G)(-(x+\varepsilon^*))$, so that $\vt(G)$ does not first order stochastically dominate $\vt(F)$.}

Define $F(x^-)\equiv\sup_{y<x} F(y)$.  We show that $\vt(F)(x) = 1-F((-x)^-)$. As $\vt(F)$ is clearly non-decreasing and satisfies $\{0,1\}\subseteq\overline{F(\mathbb{R})}\subseteq [0,1]$, we only need to show that $\vt(F)$ is right continuous. Observe that $F(x^-)$ is left continuous, by definition; that is, for any $x$ and $\varepsilon<0$, by taking $y<x$ arbitrarily close to $x$, we may choose $F(y) \geq F(x)-\varepsilon$, and as $F$ is nondecreasing, the result follows.  Therefore, $F((-x)^-)$ is right continuous and so is $\vt(F)$. 

We show that $\tau\circ \tau$ is the identity map.
Let $F\in\mathcal{F}$ and $x\in\mathbb{R}$.  Then $\tau(\tau(F))(x) = 1-\sup_{z:z<-x}[1-\sup_{y:y<-z}F(y)]=\inf_{z:z<-x}\sup_{y:y<-z}F(y)=\inf_{z:z>x}\sup_{y:z>y}F(y)$, where the last equality follows as $z$ is a dummy variable.  Observe now that if $z>x$, $\sup_{y:y<z}F(y)\geq F(x)$, as $x<z$.  So, $\tau(\tau(F))(x)\geq F(x)$.  On the other hand, since $F$ is right  continuous, for every $\varepsilon > 0$, there is $z>x$ for which $F(z)<F(x)+\varepsilon$, so that for any such $z$, $\sup_{y:y<z}F(y)\leq F(x)+\varepsilon$, since $F$ is nondecreasing.  Therefore, $\inf_{z:z>x}\sup_{y:z>y}F(y)\leq F(x)+\varepsilon$, and we can conclude that $\tau(\tau(F))(x)=F(x)$.

\textbf{Part 2.}
Let $\rho$ satisfies meet separability, translation equivariance, and upper semicontinuity, and define $\rho^*(F)=-\rho(\vt(F))$. We show that $\rho^*$ satisfies join separability, translation equivariance, and lower semicontinuity, and that $\rho(F)=-\rho^*(\vt(F))$ for all $F$.

To show that $\rho^*$ satisfies join separability, let $F,G\in\F$.
Then $\rho^*(F \joinfosd G)=-\rho(\vt(F\joinfosd G))=-\rho(\vt(F)\meetfosd \vt(G))$, and $\rho^*(F)\vee\rho^*(G)=(-\rho(\vt(F)))\vee (-\rho(\vt(G)))=-(\rho(\vt(F))\wedge \rho(\vt(G)))$. Because $\rho$ satisfies meet separability, $\rho(\vt(F)\meetfosd \vt(G))=\rho(\vt(F))\wedge \rho(\vt(G))$. 
It follows that $\rho^*(F \joinfosd G)=\rho^*(F)\vee\rho^*(G)$

To show that $\rho^*$ satisfies translation equivariance, let $F\in\F$ and $b\in\mathbb{R}$, and define $g(x) = x+b$. Then $\vt(F\circ g^{-1})=\vt(F)\circ g$, and therefore $\rho^*(F\circ g^{-1})=-\rho(\vt(F\circ g^{-1}))=-\rho(\vt(F)\circ g)$. Because $\rho$ satisfies translation equivariance, it follows that $-\rho(\vt(F)\circ g)=-(\rho(\vt(F))-b)$ and therefore that $-(\rho(\vt(F))-b)=-\rho(\vt(F))+b=\rho^*(F)+b$. Thus $\rho^*(F\circ g^{-1})=-\rho^*(F)+b$.

To show that $\rho^*$ satisfies lower semicontinuity, let $x^*\in\mathbb{R}$ and let $F_n$ be a sequence such that $\rho^*(F_n) \leq x^*$ for all $n$ and such that $F_n \rightarrow F$.  
Because $x^* \geq \rho^*(F_n)=-\rho(\vt(F_n))$, it follows that $\rho(\vt(F_n))\geq -x^*$. Because $\rho$ satisfies upper semicontinuity, and because  $\vt(F_n)\rightarrow \vt(F)$, it follows that $\rho(\vt(F))\geq -x^*$. Consequently, $\rho^*(F)=-\rho(\vt(F))\leq x^*$.

Lastly, we show that $\rho(F)=-\rho^*(\vt(F))$. By the definition $\rho^*(F)=-\rho(\vt(F))$. Consequently, $-\rho^*(\vt(F))=\rho(\vt(\vt((F)))$. Because $F=\vt(\vt(F))$ for all $F$, it follows that $-\rho^*(\vt(F))=\rho(F)$.

\textbf{Part 3.}
First we prove that statement (1) implies statement (2). Let $\rho$ satisfies the three axioms. Consequently, $\rho^*$ satisfies join separability, translation equivariance, and lower semicontinuity. By theorem~\ref{3way}, it follows that there is a $\vg\in\VG$ for which $\rho^*=\rho^{\vg}$. Because, for all $F$, $\rho(F)=-\rho^*(\vt(F))$, it follows that $\rho(F)=-\rho^*(\vt(F))=-\min\{t:\vt(F)(x)\geq \vg(x-t)\textnormal{ for all }x\}$. 

Because $-\min\{t:\vt(F)(x)\geq \vg(x-t)\textnormal{ for all }x\}=\max\{-t:\vt(F)(x)\geq \vg(x-t)\textnormal{ for all }x\}$, and because $\vt(F)=1-F((-x)^-)$, it follows that  $\rho(F)=\max\{t:1-F((-x)^-)\geq \vg(x+t)\textnormal{ for all }x\}$.
That $1-F((-x)^-)\geq \phi(x+t)$ for all $x$ obviously implies that $1-\phi(x+t)\geq F((-x)^-)$ for all $x$.  Working with $x$ instead of $-x$, we have that  
\begin{equation}\label{eq:rho}\rho(F)=\max\{t:1-\vg(t-x)\geq F(x^-)\textnormal{ for all }x\}.\end{equation}

Let $\vh_\vg:\mathbb{R}\rightarrow\mathbb{R}$ be defined as $\vh_\vg(x) = 1 - \sup_{y:y<-x}\vg(y)$. We show that $\vh_\vg\in\VH$.
To show that $\vh_\vg$ is non-decreasing, let $x,y\in\mathbb{R}$ such that $x\geq y$. Then $-x\leq-y$ implies that $\{z:z<-x\}\subseteq\{z:z<-y\}$ and therefore that $\vh_\vg(x) = 1-\sup_{z:z<-x}\vg(z) \geq 1-\sup_{z:z<-y}\vg(z) =\vh_\vg(y)$. 
{To show that $\vh_\vg$ is right-continuous, note that, by construction, $\sup_{y:y<x} \vg(x)$ is left-continuous, and therefore that $\sup_{y:y<x} \vg(-x)=\sup_{y:y<-x} \vg(x)$ is right continuous. Consequently, $1-\sup_{y:y<-x} \vg(x)=\vh_\vg(x)$ is right continuous.}
To show that there exists $x$ such that $\vh_\vg(x)=1$, note that there exists $z$ such that $x<z$ implies that $\vg(x)=0$. Then $\vh_\vg(-z) = 1 - \sup_{y:y<z}\vg(y) = 1$. 
To show that there exists $x$ such that $\vh_\vg(x)<1$, note that there exists $z$ such that $x\geq z$ implies that $\vg(x)>0$. Let $w>z$. Then $\vh_\vg(-w) = 1 - \sup_{y:y<w}\vg(y) = 1$. Because $z<w$, $\sup_{y:y<w}\vg(y) \geq \vg(z)$, which implies that $\vh_\vg(-w) = 1 - \sup_{y:y<w}\vg(y) < 1$. 
To show that $\lim_{x\rightarrow-\infty}\vh_\vg(x)>0$, note that $\lim_{x\rightarrow-\infty}\sup_{y:y<-x}\vg(y)=\lim_{x\rightarrow\infty}\sup_{y:y<x}\vg(y)=\lim_{x\rightarrow\infty}\vg(x)<1$. Then  $\lim_{x\rightarrow-\infty}\vh_\vg(x) = 1 - \lim_{x\rightarrow\infty}\vg(x)>0$.

Finally, we claim that $\vh_\vg(x-t)\geq F(x)$ for every $x$ if and only if $1-\vg(t-x)\geq F(x^-)$ for every $x$.

First, let $1-\vg(t-x)\geq F(x^-)$ for every $x$.   Then for every $z>x$, we know that $1-\vg(t-z)\geq F(z^-)\geq F(x)$.  Consequently, for every $y>z$, we also have $1-\vg(t-y)\geq F(x)$ as $y>z$ implies $y>x$.  Note that $\vh_\vg(x) = 1 - \sup_{y:y<-x}\vg(y) = \inf_{y:y<-x}(1-\vg(y))$.
Because $t-y<t-z$, it follows that $\vh_\vg(z-t)=\inf_{t-y:t-y<t-z}(1-\vg(t-y))\geq F(x)$.  Finally, since $\vh_\vg$ is right continuous, we establish that $\vh_\vg(x-t)\geq F(x)$.

Conversely, suppose by means of contradiction  that $\vh_\vg(x-t)\geq F(x)$ for every $x$ but that there is an $x$ for which $1-\vg(t-x) < F(x^-)$.  Let $\varepsilon > 0$.  Then there is $z<x$ for which for all $y\in [z,x)$, we have $F(y)\geq F(x^-)-\varepsilon$.  By choosing  $\varepsilon>0$ small enough, we get a pair of numbers $a,b$ for which for all $y\in [z,x)$, we have $1-\vg(t-y)\leq 1-\vg(t-x)\leq a < b \leq F(y)$, where the first inequality is because $\vg$ is non-decreasing.  Now, for any $y\in (z,x)$, we have $1-\vg(t-z)\leq 1-\vg(t-y)\leq a$ by nondecreasingness.  Consequently, $\inf_{t-y:t-y<t-z}(1-\vg(t-y))\leq a$.
However, $\vh_\vg(z-t)  = \inf_{t-y:t-y<t-z}(1-\vg(t-y))\geq  F(z) \geq b > a$, a contradiction.

We may conclude, using this fact and equation~\eqref{eq:rho}, that $$\rho(F)=\max\{t:F(x) \leq \vh_\vg(x-t)\textnormal{ for all }x\}.$$

\textbf{Part 4.}
Next we prove that statement (1) implies statement (3). Let $\rho$ satisfy the three axioms. Consequently, $\rho^*$ satisfies join separability, translation equivariance, and lower semicontinuity. By theorem~\ref{3way}, it follows that there is a $\vc\in\VC$ for which $\rho^*=\rho^{\vc}$. 

Because, for all $F$, $\rho(F)=-\rho^*(\vt(F))$, it follows that $\rho(F)=-\sup_{\alpha\in(0,1)}\left(\rho_{\alpha}(\vt(F))-c(\alpha)\right)$. 
It follows that $\rho(F)=\inf_{\alpha\in(0,1)}\left(c(\alpha)-\rho_{\alpha}(\vt(F))\right)$.
By the definition of $\rho_{\alpha}$ and because $\vt(F)(x) = 1-F((-x)^-)$, it follows that $\rho(F)=\inf_{\alpha\in(0,1)}\left(c(\alpha)-\inf\{x\in\mathbb{R}:1-F((-x)^-)\geq\alpha\}\right)$.

Rearranging, we have $\rho(F)=\inf_{\alpha\in(0,1)}\left(c(\alpha)+\sup\{-x\in\mathbb{R}:1-F((-x)^-)\geq\alpha\}\right)$, and therefore that $\rho(F)=\inf_{\alpha\in(0,1)}\left(c(\alpha)+\sup\{x\in\mathbb{R}:F(x^-)\leq 1-\alpha\}\right)$.

Define $\vd(\vc)(\alpha)\equiv -\vc(1-\alpha)$. We show that $\vd(\vc)\in\VD$. That $\vd(\vc)$ is non-decreasing follows from the fact that $\vc(1-\alpha)$ is non-increasing. To show that $\vd(\vc)$ is upper semicontinuous, note that $\vc(\alpha)$, and consequently $\vc(1-\alpha)$ are lower semicontinuous. To show that there is  $b\in(0,1)$ such that $\vd(\vc)=-\infty$, note that there exists $b'\in(0,1)$ such that $\vc(b')=\infty$. This implies that $1-b\in(0,1)$ and that $\vd(\vc)(1-b')=-\infty$.

By substituting $\vd(\vc)$ for $\vc$ we get $\rho(F)=\inf_{\alpha\in(0,1)}\left(\sup\{x\in\mathbb{R}:F(x^-)\leq \alpha\}-\vd(\vc)(\alpha)\}\right)$.
To show that $\rho(F)=\inf_{\alpha\in(0,1)}\left(\rho_\alpha^+(F)-\vd(\vc)(\alpha)\}\right)$, it is sufficient to establish that  that $\sup\{x\in\mathbb{R}:F(x^-)\leq \alpha\}= \sup\{x\in\mathbb{R}:F(x^)\leq \alpha\}$.

Observe that $F(x^-)=\sup_{y<x}F(y)\leq F(x)$, as $F$ is nondecreasing.  So $F(x) \leq \alpha$ implies $F(x^-)\leq \alpha$, whereby $\sup\{x:F(x^-)\leq \alpha\}\geq \sup\{x:F(x)\leq \alpha\}$.  Since $F(x^-)$ is left continuous and nondecrasing, it is lower semicontinuous, and $\sup\{x:F(x^-)\leq \alpha\}$ is actually achieved at a maximum, say $x^*$.  By definition of $F((x^*)^-)$, for any $y<x^*$, we have $F(y) \leq \alpha$.  Consequently, $\sup\{x:F(x)\leq\alpha\}\geq x^*$, establishing the other direction.

\textbf{Part 5.}
Let $\rho\equiv\rho^\vh$ for some $\vh\in\VH$. We show that $\rho$ satisfies meet separability, translation equivariance, and upper semicontinuity.

To show that $\rho$ satisfies meet separability, let $F,G\in\F$, and assume that $\rho(F)\leq\rho(G)$. Let $z=\rho(F)$.
Then $z=\max\{t:F(x)\leq \vh(x-t)\textnormal{ for all }x\}\leq\max\{t:G(x)\leq \vh(x-t)\textnormal{ for all }x\}$, which implies that $F(x)\leq \vh(x-z)$ and $G(x)\leq \vh(x-z)$ for all $x$. Consequently $(F\meetfosd G)(x)\leq \vh(x-z)$ for all $x$, which implies that $\rho(F\meetfosd G)\geq z$.
Let $\varepsilon>0$. Then $z+\varepsilon>\rho(F)$ which implies that, for some $x$, $F(x) > \vh(x-(z+\varepsilon))$. By construction, $(F\meetfosd G)(x)\geq F(x)$ for all $x$, which implies that $(F\meetfosd G)(x) > \vh(x-(z+\varepsilon))$ for some $x$, and therefore that $\rho(F\meetfosd G)\leq z+\varepsilon$. Since $\varepsilon$ was arbitrary, this implies that $\rho(F\meetfosd G) = z = \min\{\rho(F),\rho(G)\}$.

To show that $\rho$ satisfies translation equivariance, let $b\in\mathbb{R}$ and $F\in\F$, and define $h(x)=x+b$. Then $\rho(F\circ h^{-1}) =
\max\{t:(F\circ h^{-1})(x)\leq \vh(x-t)\textnormal{ for all }x\} =
\max\{t:F(x)\leq \vh(h(x)-t)\textnormal{ for all }x\} =
\max\{t+b:F(x)\leq \vh(x-t)\textnormal{ for all }x\} =
\max\{t:F(x)\leq \vh(x-t)\textnormal{ for all }x\} + b =
\rho(F) + b$.

 To see that $\rho$ satisfies upper semicontintuity, let $F\in\mathcal{F}$, $z\in\mathbb{R}$, and $\{F_n\}\subseteq \mathcal{F}$ for which $F_n \rightarrow F$.  Suppose that $\rho(F_n)\geq z$ for all $n$.  We claim that $\rho(F)\geq z$.  To this end, observe that for all $n$ and all $x\in\mathbb{R}$, $F_n(x)\leq \phi(x-z)$.  Now, for all continuity points $y$ of $F$, we have $F_n(y)\rightarrow F(y)$.  Consequently for all continuity points $y$ of $F$, $F(y)\leq \phi(y-z)$.  Because $F$ is nondecreasing, the set of continuity points is dense, and so we may conclude that $F(x) \leq \phi(x-z)$ for all $x$ as each of $F$ and $\phi$ are right continuous.  Therefore, $\rho(F)\geq z$.

\textbf{Part 6.}
We first show that the upper quantile $\rho^+_\alpha$ satisfies meet separability, translation equivariance, and upper semicontinuity.

To see that $\rho^+_\alpha$ satisfies meet separability, let $F,G\in\F$ such that $\rho^+_\alpha(F)\geq\rho^+_\alpha(G)$, and let $z=\rho^+_\alpha(G)$. 
Because $z=\rho^+_\alpha(G)$ it follows that $z=\sup\{x\in\mathbb{R}:G(x)\leq\alpha\}$.
This implies that, for all $\varepsilon>0$, $G(z+\varepsilon)>\alpha$, and consequently that
$(F\meetfosd G)(z+\varepsilon)>\alpha$.
Therefore, $\rho^+_\alpha(F\meetfosd G)\leq z$.
Next, for all $\varepsilon>0$, 
$z-\varepsilon<\rho^+_\alpha(G)\leq\rho^+_\alpha(F)$ which implies that $G(z-\varepsilon)\leq\alpha$ and $F(z-\varepsilon)\leq\alpha$. It follows that $(F\meetfosd G)(z-\varepsilon)\leq\alpha$, and therefore that $z\leq\rho^+_\alpha(F\meetfosd G)$.  Thus  $\rho^+_\alpha(F\meetfosd G)= z = \min\{\rho^+_\alpha(F),\rho^+_\alpha(G)\}$.

To see that $\rho^+_\alpha$ satisfies translation equivariance, let $b\in\mathbb{R}$ and $F\in\F$, and define $h(x)=x+b$. Then $\rho^+_\alpha(F\circ h^{-1}) = \sup\{x \in \mathbb{R}: (F\circ h^{-1})(x) \leq \alpha\} = \sup\{h(x) \in \mathbb{R}: F(x) \leq \alpha\}=h(\sup\{x \in \mathbb{R}: F(x) \leq \alpha\})=h(\rho^+_\alpha(F))=\rho^+_\alpha(F)+b$.

{That $\rho^+_\alpha$ satisfies upper semicontinuity follows from \citet[Theorem 1]{cpc:2007}. To see this, note that $\sup\{x\in\mathbb{R}:F(x)\leq\alpha\}=\inf\{x\in\mathbb{R}:F(x)>\alpha\}$.}

Next, let $\rho\equiv\rho^\vd$ for some $\vd\in\VD$. We show that $\rho$ satisfies meet separability, translation equivariance, and upper semicontinuity.

To show that $\rho$ satisfies meet separability, let $F,G\in\F$.  
For $\alpha\in(0,1)$, the fact that $\rho^+_\alpha$ satisfies meet separability implies that $\left(\rho^+_{\alpha}(F \meetfosd G)-\vd(\alpha)\right)=\min_{E\in\{F,G\}}\left(\rho^+_{\alpha}(E)-\vd(\alpha)\right)$.
 Consequently, $\rho(F \meetfosd G)=\inf_{\alpha\in(0,1)}\min_{E\in\{F,G\}}\left(\rho^+_{\alpha}(E)-\vd(\alpha)\right)=\min_{E\in\{F,G\}}\inf_{\alpha\in(0,1)}\left(\rho^+_{\alpha}(E)-\vd(\alpha)\right)=\min\{\rho(F),\rho(G)\}$.

To show that $\rho$ satisfies translation equivariance, let $b\in\mathbb{R}$ and $F\in\mathcal{F}$, and define $h(x)=x+b$. 
Because $\rho_\alpha^+$ satisfies translation equivariance, $\rho^+_{\alpha}(F\circ h^{-1})=\rho^+_{\alpha}(F)+b$.
Consequently, $\rho(F\circ h^{-1})=\inf_{\alpha\in(0,1)}\left(\rho^+_{\alpha}(F\circ h^{-1})-\vd(\alpha)\right) = \inf_{\alpha\in(0,1)}\left(\rho^+_{\alpha}(F)+b-\vc(\alpha)\right)=\inf_{\alpha\in(0,1)}\left(\rho^+_{\alpha}(F)-\vd(\alpha)\right)+b$. 
Consequently, $\rho(F\circ h^{-1})=\rho(F)+b$. 

Upper semicontinuity of $\rho$ follows from the fact that a pointwise infimum of upper semicontinuous functions is itself upper semicontinuous, together with the fact that each $\rho_{\alpha}^+$ is upper semicontinuous.  See \emph{e.g.} \citet{berge}, Theorem 3 of p. 76 for the result framed in terms of lower semicontinuity.
\end{proof}

\begin{proof}[Proof of Theorem~\ref{thm:ordinal}] That the axioms must be satisfied by any function $\rho$ satisfying the properties of Theorem~\ref{thm:dual} is straightforward.  We establish the converse here.

We first establish that for every $F\in \F$, there is $x\in\mathbb{R}$ for which $F_x \succeq F$.  Observe first that by ordinal meet separability, if $F \geqfosd G$, then $F \succeq G$.  Second, observe that if $x>y$, then $F_x \succ F_y$; for, suppose not, then by completeness, $F_y \succeq F_x$.  By ordinal translation invariance $F_{y-x}\succeq F_0$.  By ordinal translation invariance again, for any $n\in\mathbb{N}$, $F_{n(y-x)}\succeq F_{(n-1)(y-x)}$.  Then by transitivity, for any $n\in\mathbb{N}$, $F_{n(y-x)}\succeq F_0$.  Observe also that by lower unboundedness, there is $z\in\mathbb{R}$ for which $F_0 \succ F_z$.  Thus for any $n\in\mathbb{N}$, $F_{n(y-x)}\succ F_z$, by weak order.  For $n$ large this contradicts the monotonicity property ($F \geqfosd G$ implies $F\succeq G$).

In particular, for every $F$ for which there exists $x\in\mathbb{R}$ for which $F(x) = 1$, we know that $F_x \geqfosd G$, and so $F_x \succeq F$.  Suppose now that there is no such $x\in\mathbb{R}$, so that $F(x)<1$ for every $x\in\mathbb{R}$.

We want to show that there is $x\in\mathbb{R}$ for which $F_x \succeq F$; by means of contradiction and invoking completeness, suppose that for all $x\in\mathbb{R}$,  $F \succ F_x$.

By ordinal meet-separability, for every $x\in\mathbb{R}$, $F_x \sim (F_x \meetfosd F)$.  By ordinal translation invariance, for every $x\in\mathbb{R}$, $F_0 \sim (F_0 \meetfosd (F-x))$.  Now, consider $(F_0 \meetfosd F-n)\joinfosd F_{-1}$; because $(F_0 \meetfosd F-n)\joinfosd F_{-1}\geqfosd (F_0\meetfosd(F-n))$, we obtain by transitivity that $(F_0 \meetfosd F-n)\joinfosd F_{-1} \succeq F_0$.  Now observe that as $n\rightarrow \infty$, we have $(F_0\meetfosd F-n)\joinfosd F_{-1}\rightarrow F_{-1}$.  Conclude by ordinal upper semicontinuity that $F_{-1} \succeq F_0$, but we have shown that $F_0 \succ F_{-1}$ as $0 > -1$, a contradiction.

Apply Rader's representation Theorem \citep{rader1963,richter1980} to find an upper semicontinuous function $V:\F\rightarrow\mathbb{R}$ representing $\succeq$.  That this is possible follows because the set of probability distributions on $\mathbb{R}$ endowed with the topology of weak convergence has a countable base---it is a separable space (Theorem 15.12 of \citet{ab:2006}), which implies that it has a countable base (p. 237 of \citet{billingsley}).

For any $F\in\F$, we know that there are $\overline{x}$,$\underline{x}\in\mathbb{R}$ for which $F_{\overline{x}}\succeq F \succ F_{\underline{x}}$; the lower bound follows from lower unboundedness.  We claim that there is $x\in\mathbb{R}$ for which $F\sim F_{x}$.

To see this, let $x^* \equiv \inf\{x:F_x \succeq F\}$.  The preceding bounds imply that $x^*$ exists and is real-valued.  Furthermore, by ordinal upper semicontinuity, we have $F_{x^*} \succeq F$.  We claim that $F_{x^*}\sim F$.  Suppose by means of contradiction that $F_{x^*}\succ F$ and observe that this implies that $V(F_{x^*})>V(F)$.  Further by definition of $x^*$, for any $y<x^*$, we have $V(F)>V(F_y)$.  Now, let $b,b'\in\mathbb{R}$ for which $b'<b$.  Observe that $V(F_{x^*})>V(F)>V(F_{x^*+(b'-b)})$.  By ordinal translation invariance, $V(F_{x^*}+b)>V(F+b)>V(F_{x^*}+b')>V(F+b')$.  For each $b\in \mathbb{R}$, let $a(b)\in \mathbb{Q}\cap (V(F+b),V(F_{x^*}+b))$.\footnote{This does not require the axiom of choice.}  Because $b> b'$ implies $(V(F+b),V(F_{x^*}+b))\cap (V(F+b'),V(F_{x^*}+b'))=\varnothing$, we have constructed a one to one map from $\mathbb{R}$ into $\mathbb{Q}$, which is an obvious contradiction.

Now, for each $F$, define $\rho(F)=x$, where $F \sim F_x$.  Clearly this choice of $x$ is unique; as if $x>y$, $F_x \succ F_y$.  That $\rho$ represents $\succeq$ follows from the fact that $F_x \succeq F_y$ iff $x \geq y$ and the fact that $F_{\rho(F)}\sim F \succeq G\sim F_{\rho(G)}$.\footnote{That $F_x \succeq F_y$ iff $x \geq y$ follows as $x \geq y$ implies $F_x \geqfosd F_y$, so that $F_x \succeq F_y$ and $x > y$ implies $F_x \succ F_y$.}  We claim that $\rho$ satisfies the properties of Theorem~\ref{thm:dual}.  Meet-separability is obvious.  Translation invariance follows as by ordinal translation invariance, we have $F_{\rho(F)}\sim F$ implies that $F_{\rho(F)+b}=F_{\rho(F)}+b\sim F+b$.  Finally upper semicontinuity follows thusly.  Suppose $\rho(F_n)\geq x^*$ and that $F_n \rightarrow F$.  Observe that by translation invariance there is $G$ for which $\rho(G)=x^*$.  Consequently, $F_n \succeq G$ as $\rho$ represents $\succeq$.  By ordinal upper semicontinuity, $F \succeq G$, which implies that $\rho(F)\geq \rho(G)=x^*$.  \end{proof}

\bibliographystyle{ecta}
\bibliography{quantilebib}

\end{document}